\documentclass[conference]{IEEEtran}
\IEEEoverridecommandlockouts

\usepackage{silence}
\usepackage{cite}

\usepackage{amsmath,amssymb,amsthm,fixmath}

\usepackage{color,colortbl}
\usepackage{xcolor}

\usepackage[inline]{enumitem}

\usepackage{siunitx}
\DeclareSIUnit{\dBm}{dBm}

\usepackage{graphicx}
\usepackage[labelformat=simple]{subcaption}

\usepackage{epstopdf}
\usepackage{lipsum}

\usepackage{cuted}
\usepackage{multirow,booktabs}
\usepackage{diagbox}
\usepackage{makecell}
\usepackage{hhline}
\usepackage{array} 
\newcolumntype{x}{!{\vrule width 2px}}
\newcolumntype{y}{!{\vrule width 1.5px}}

\usepackage{optidef}

\usepackage[hidelinks]{hyperref}

\usepackage[shortcuts,acronym,automake]{glossaries}
\makeglossaries
\newacronym{6g}{6G}{Sixth Generation}
\newacronym{awgn}{AWGN}{additive white Gaussian noise}
\newacronym{bcd}{BCD}{block coordinate descent}
\newacronym{bec}{BEC}{binary erasure channel}
\newacronym{bler}{BLER}{block error rate}
\newacronym{blec}{BLEC}{block erasure channel}
\newacronym{bri}{BRI}{biregular irreducible}
\newacronym{clt}{CLT}{central limit theorem}
\newacronym{cp}{CP}{control plane}
\newacronym{cdf}{CDF}{cumulative distribution function}
\newacronym{csi}{CSI}{channel state information}
\newacronym{csit}{CSIT}{channel state information at transmitter}
\newacronym{dft-s-ofdm}{DFT-s-OFDM}{Discrete Fourier Transform-spread-OFDM}
\newacronym{fbl}{FBL}{finite blocklength}
\newacronym{gan}{GAN}{generative adversarial network}
\newacronym{ibl}{IBL}{infinite blocklength}
\newacronym{ldpc}{LDPC}{low-density parity-check}
\newacronym{lfp}{LFP}{leakage-failure probability}
\newacronym{ls}{LS}{least squares}
\newacronym{mac}{MAC}{medium access control}
\newacronym{mcs}{MCS}{modulation and coding scheme}
\newacronym{mimo}{MIMO}{multi-input multi-output}
\newacronym{ml}{ML}{maximum likelihood}
\newacronym{mm}{MM}{Minorize-Maximization}
\newacronym{noma}{NOMA}{non-orthogonal multi-access}
\newacronym{nom}{NOM}{non-orthogonal multiplexing}
\newacronym{ofdm}{OFDM}{orthogonal frequency-division multiplexing}
\newacronym{ofdma}{OFDMA}{orthogonal frequency-division multiple access}
\newacronym{oma}{OMA}{orthogonal multiple access}
\newacronym{papr}{PAPR}{Peak-to-Average Power Ratio}
\newacronym{pdf}{PDF}{probability density function}
\newacronym{per}{PER}{packet error rate}
\newacronym{phy}{PHY}{physical}
\newacronym{pld}{PLD}{physical layer deception}
\newacronym{pls}{PLS}{physical layer security}
\newacronym{prb}{PRB}{physical resource block}
\newacronym{psk}{PSK}{phase-shift keying}
\newacronym{sic}{SIC}{successive interference cancellation}
\newacronym{sinr}{SINR}{signal-to-interference-and-noise ratio}
\newacronym{snr}{SNR}{signal-to-noise ratio}
\newacronym{tdma}{TDMA}{time-division multiple access}
\newacronym{up}{UP}{user plane}
\newacronym{urllc}{URLLC}{ultra-reliable low-latency communication}
\newacronym{xurllc}{xURLLC}{next-generation URLLC}

\usepackage{tcolorbox}
\tcbuselibrary{many}
\newtheorem{theorem}{Theorem}
\newtheorem{lemma}{Lemma}
\newtheorem{corollary}{Corollary}

\usepackage[norelsize,linesnumbered,ruled]{algorithm2e}
\SetKwRepeat{Do}{do}{while}
\makeatletter
\newcommand{\removelatexerror} {\let\@latex@error\@gobble}
\makeatother

\usepackage{flushend}

\newcommand{\superscript}[1]{^{\mathrm{#1}}}
\newcommand{\subscript}[1]{_{\mathrm{#1}}}
\newcommand{\diff}{\text{d}}
\newcommand{\comb}[2]{\begin{pmatrix}#1\\#2\end{pmatrix}}

\usepackage[normalem]{ulem}

\usepackage[textsize=tiny,colorinlistoftodos]{todonotes}
\makeatletter
\define@key{todonotes}{bh}[]{
	\setkeys{todonotes}{author=\textbf{Bin}, color=lime!30}}%
\define@key{todonotes}{yz}[]{
	\setkeys{todonotes}{author=\textbf{Yao}, color=blue!30}}%
\define@key{todonotes}{rs}[]{
	\setkeys{todonotes}{author=\textbf{Rafael}, color=green!30}}%
\makeatother

\tikzstyle{note}=[rectangle, minimum width=3cm, draw = none, fill = none, minimum width = 1.5cm, anchor=center, align=left]
\tikzstyle{block}=[rectangle, draw, line width=1pt, fill = none, minimum width = 1cm, minimum height = 0.75cm, anchor=center, inner sep = 0.5mm, align=center]
\tikzstyle{arrow} = [thick,->,>=stealth]

\newif\ifreviewmode
\reviewmodetrue 

\ifreviewmode
\else
  \renewcommand{\todo}[1]{} 
\fi

\begin{document}

\title{Confusions and Erasures of Error-Bounded Block Decoders with Finite Blocklength}

\author{
	\IEEEauthorblockN{
		Bin~Han\IEEEauthorrefmark{1}, 
        Yao~Zhu\IEEEauthorrefmark{2}, 
        Rafael~F.~Schaefer\IEEEauthorrefmark{3}, 
		Giuseppe~Caire\IEEEauthorrefmark{4},\\ 
        Anke~Schmeink\IEEEauthorrefmark{2}, 
        H.~Vincent~Poor\IEEEauthorrefmark{5}, 
        and~Hans~D.~Schotten\IEEEauthorrefmark{1}\IEEEauthorrefmark{6} 
	}
	\IEEEauthorblockA{
		\IEEEauthorrefmark{1}RPTU University Kaiserslautern-Landau,
        \IEEEauthorrefmark{2}RWTH Aachen University,
		\IEEEauthorrefmark{3}TU Dresden,
		\IEEEauthorrefmark{4}TU Berlin,\\
		\IEEEauthorrefmark{5}Princeton University,
		\IEEEauthorrefmark{6}German Research Center for Artificial Intelligence (DFKI GmbH)
	}
	\thanks{
    The work of B. Han and H. D. Schotten was in part supported by the Federal Ministry of Research, Technology and Space (BMFTR) of Germany in the project Open6GHub+ (16KIS2406). The work of Y. Zhu and A. Schmeink was in part supported by BMFTR in the project 6GEM+ (16KIS2409K). The work of R. F. Schaefer was in part supported by BMFTR in the project 6G-life (16KIS2413K), and by the German Research Foundation (DFG) in the Cluster of Excellence EXC
2050/2 ``Centre for Tactile Internet with Human-in-the-Loop'' (CeTI) of TU Dresden (390696704). The work of G. Caire was in part supported by BMFTR in the project 6G-RIC (16KISK030). The work of H. V. Poor was supported in part by an Innovation Grant from Princeton NextG. B. Han (bin.han@rptu.de) and Y. Zhu (yao.zhu@inda.rwth-aachen.de) are the corresponding authors.}
}

\maketitle

\begin{abstract}
	This paper investigates two distinct types of block errors -- undetected errors (confusions) and erasures -- in \ac{awgn} channels with error-bounded block decoders operating in the \ac{fbl} regime. While \ac{bler} is a common metric, it does not distinguish between confusions and erasures, which can have significantly different impacts in cross-layer protocol design, despite upper-layer protocols universally assuming \ac{phy} errors manifest as packet erasures rather than undetected corruptions -- an assumption lacking rigorous \ac{phy}-layer validation. We present a systematic analysis of confusions and erasures under \ac{bler}-constrained \ac{ml} decoding. Through sphere-packing analysis, we provide analytical bounds for both block confusion and erasure probabilities, and derive the sensitivities of these bounds to blocklength and \ac{snr}. To the best of our knowledge, this is the first study on this topic in the \ac{fbl} regime. Our findings provide theoretical validation for the block erasure channel abstraction commonly assumed in \ac{mac} and network layer protocols, confirming that, for practical \ac{fbl} codes, block confusions are negligible compared to block erasures, especially at large blocklengths and high \ac{snr}.
\end{abstract}

\begin{IEEEkeywords}
    Finite blocklength, block erasing channel, bounded decoder
\end{IEEEkeywords}

\glsresetall

\section{Introduction}\label{sec:intro}
Since decades, \ac{bler} has been widely used as a performance metrics for digital communication systems that use block coding. Especially, it has been becoming increasingly important in the context of recent developments of \ac{urllc} and \ac{xurllc}, for they generally work in the \ac{fbl} regime, where the traditional concept of channel capacity in the Shannon sense fails and lossless transmission is considered impractical in general~\cite{PPV2010channel}.

However, the \ac{bler} does not provide a complete picture of the transmission performance, as it counts both undetected block decoding errors and block erasures. More specifically, the former phenomenon occurs when the decoder confuses the true sent codeword for a wrong one with such high confidence that it fails to correct or detect such an error. In this paper, we call this type of error a block confusion. The latter, to the contrary, occurs when a bounded decoder cannot give any estimate for the received codeword with sufficient confidence, and therefore rejects all candidate codewords, issuing an ``erasure/loss'' as output. 

Block erasures, and the associated \acp{blec}, have been extensively studied in the \ac{ibl} regime. Theoretical bounds on erasure probability have been obtained for bounded decoders~\cite{Forney1968exponential,Merhav2008error,SM2010exact}. Meanwhile, little effort has been reported in the \ac{fbl} regime to distinguish block erasures from block confusions, 
despite networking protocols universally modeling \ac{phy} failures as erasures in retransmission schemes, \ac{mac} protocols, and reliability mechanisms. This fundamental cross-layer assumption -- that block errors at the \ac{phy} translate to detectable losses rather than silent corruptions at upper layers—remains theoretically unvalidated. 
A few existing \ac{fbl} works, such as \cite{HZS+2025semantic}, have explicitly noted this distinction, but for convenience they simply assumed the \ac{bler} to be approximately equal to the block erasure rate, i.e., they considered the block confusions negligible in comparison to the block erasures. Evidence supporting such an assumption, however, is missing, and the question of how to quantify the \ac{fbl} block erasure rate remains open.

In this work, we aim to fill this gap by providing a systematic analysis of the \ac{fbl} block erasure rate in \ac{awgn} channels under \ac{bler}-bounded \ac{ml} decoding. Our analysis provides the first rigorous \ac{phy}-layer justification for this standard networking abstraction, confirming when the erasure-only assumption holds in practical FBL systems. Following the classical information theoretic approach, we map the codewords onto a hyper-sphere in the blocklength-dimensional Euclidean space, outline the decision regions of decoder as hyper-spheres centered at the valid codewords, and transform the coding problem into a geometric one under sphere packing constraints. We derive probability bounds of block confusions and block erasures, and analyze their behavior. Our results confirm that the \ac{fbl} block confusion rate is negligibly low compared to the block erasure rate, especially at large blocklengths and high \ac{snr}.

The remainder of this paper is organized as follows. Sec.~\ref{sec:related} reviews related works on \ac{fbl} information theory, \ac{blec}, and use scenarios that essentially distinguish block confusions from block erasures. Sec.~\ref{sec:bound_analysis} outlines the design and performance of error-bounded block decoders allowing erasures, deriving the generic formulation of their block erasure probability and block confusion probability. Sec.~\ref{sec:bound_analysis} provides the bound analysis of these probabilities, and investigates behavior of bounds with blocklength and \ac{snr}. Finally, Sec.~\ref{sec:results} presents the results of our numerical experiments that support our analytical results, before Sec.~\ref{sec:conclusion} concluding this paper and giving outlooks to future study.


\section{Related Works}\label{sec:related}
The fundamental concepts and methodologies of information theory and coding theory were established in the late 1940s and 1950s, with the ground-laying work of Shannon~\cite{Shannon1949communication} and Hamming~\cite{Hamming1950error}, modeling codewords and their decision regions as hyper-spheres closely packed in Euclidean space. The same classical approach has been adopted and extensively further developed over ensuing decades, and will be followed in this paper as well.

Classical information theory mainly focuses on the asymptotic regime where blocklength approaches infinity, relying on typicality arguments and the asymptotic equipartition property to achieve capacity. For short codes with finite blocklengths, these typicality-based approximations fail to capture the non-asymptotic behavior of error probability, necessitating refined analytical frameworks.

While early finite-length analyses existed for specific coding schemes~\cite{DPT+2002finite}, systematic \ac{fbl} information theory emerged with the foundational work of Polyanskiy et al.~\cite{PPV2010channel}, who established precise \ac{bler} bounds for \ac{awgn} channels. This result was subsequently extended to fading channels~\cite{DKP2016toward} and general multi-antenna channels~\cite{YDKP2014quasi}, and later investigated for both coherent and non-coherent multiple-antenna fading channels~\cite{AY2019coherent,QK2025noncoherent}. Since then, this \ac{bler} bound has become a fundamental tool for evaluating the performance of \ac{fbl} codes across numerous applications, e.g., cooperative relay networks~\cite{HGS2015capacity}, \ac{noma} schemes~\cite{XYC+2020noma}, wireless power transfer systems~\cite{AFSA2018wireless}, etc.


Erasure is also a classic topic in communications, with foundational contributions dating back to Forney's 1968 work~\cite{Forney1968exponential}, which established error exponents for the trade-off between errors/confusions and erasures. These bounds were subsequently refined by Merhav~\cite{Merhav2008error} and Shamai et al.~\cite{SM2010exact}. These error exponent results operate in the large deviation regime, characterizing exponential decay rates as the blocklengths approaches infinity, which differs from Polyanskiy-style analysis that determines achievable rates for fixed blocklength and error probability.
In this century, channels with erasures have been extensively studied on different levels. For \acp{bec}, the authors of~\cite{DGP+2006capacity} have derived the network capacity, which can be also extended to longer blocks. Focusing on \ac{ldpc} codes in \acp{bec}, channel properties are provided in ~\cite{AKT2004extrinsic}, and the \ac{fbl} performance is analyzed in~\cite{DPT+2002finite}. For $q$-ary erasure channels, the authors of \cite{LPC2013bounds} have derived the performance bounds, and methods are proposed in \cite{MTS+2022methods} to convert error channels to pure erasure channels. For \acp{blec}, performance and bounds of channel coding is analyzed in ~\cite{GQ2006coding,Didier2006new}.

As mentioned in Sec.~\ref{sec:intro}, the distinction between block confusion, block erasure, and generic block error, can be crucial in certain application scenarios such semantic communications, which prioritize the human perception that generally and significantly distinguishes the absence of information from the distortion of information~\cite{Smithson2010ignorance}. In some applications such as~\cite{HZS+2025semantic}, erasures can be exploited while confusions shall be avoided. In some others, this preference can be reversed or differ from one packet to another~\cite{Khosravirad2025rateless}.



\section{Error-Bounded Decoding with Erasures}\label{sec:bounded_decoding}
While many conventional decoders are designed to select only the most likely feasible codeword, there are also implementations that allow options of erasure (selecting none) and/or list (listing multiple ones), as illustrated in Fig.~\ref{fig:decoding_classes}~\cite{Forney1968exponential}. In this paper, we focus on the case where no list but erasure is allowed.
\begin{figure}[!htpb]
	\centering
	\includegraphics[width=.75\linewidth]{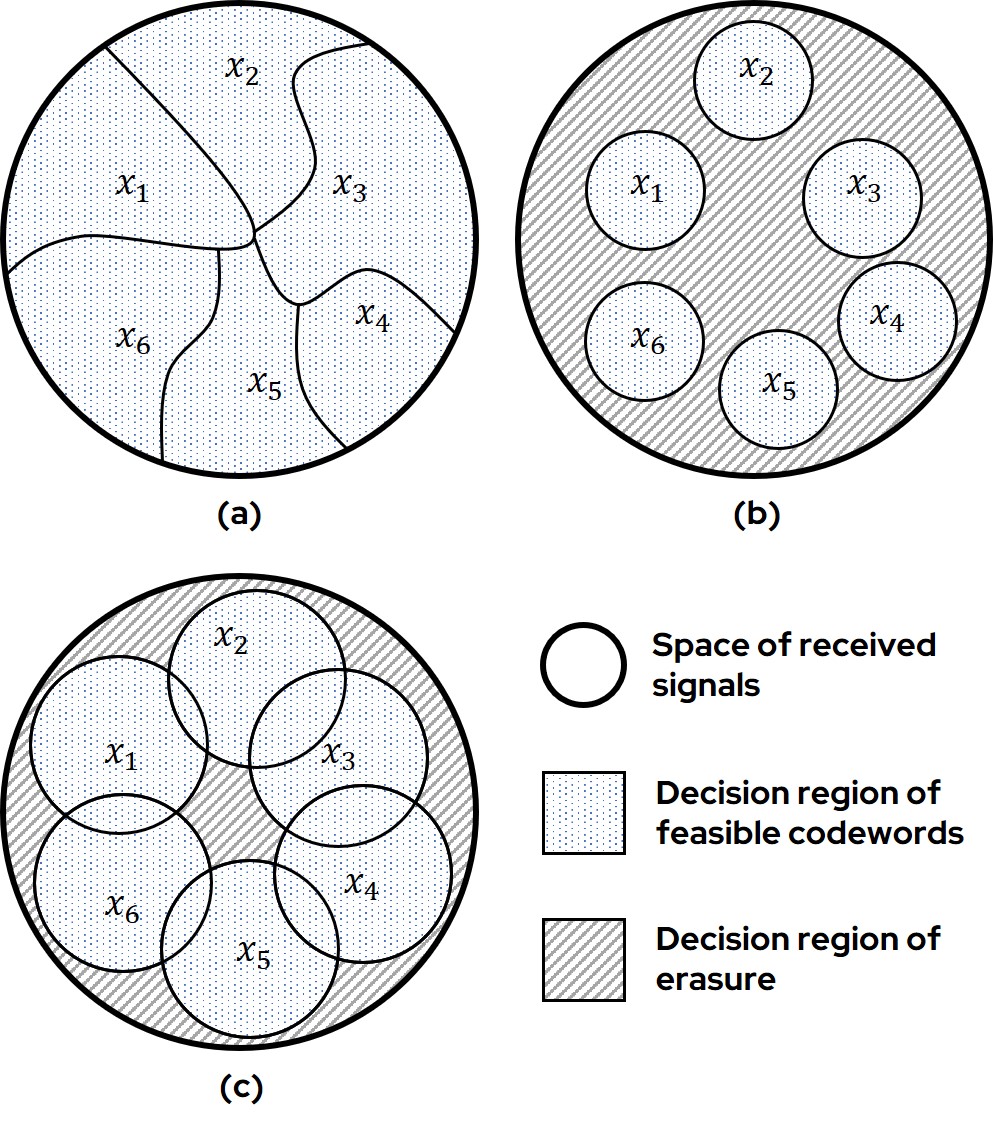}
	\caption{Schematic representation of typical decision regions: (a) ordinary decoding, (b) erasure option, and (c) list option.}
	\label{fig:decoding_classes}
\end{figure}

Consider an $M$-ary (orthogonal) codebook $\mathcal{X}$ with blocklength $n=k+r$, where $k$ is the payload length and $r$ the redundancy length. 
Now given an arbitrary code $\mathbf{x}=(x_1,x_2,\dots x_n)\in\mathcal{X}$, and an ideal channel with \ac{awgn} $\mathbf{w}\sim\mathcal{N}(0,\sigma^2I_n)$, the probability of receiving a sequence  $\mathbf{y}=\mathbf{x}+\mathbf{w}$ is given by
\begin{equation}
	P(\mathbf{y}|\mathbf{x})=\prod_{i=1}^n P(y_i|x_i)=\prod_{i=1}^n \frac{1}{\sqrt{2\pi\sigma^2}}\exp\left(-\frac{(y_i-x_i)^2}{2\sigma^2}\right)
	\label{eq:prob}
\end{equation}
So the \ac{ml} decoder of $\mathbf{x}$ is given by
\begin{equation}
	\hat{\mathbf{x}}=\arg\max_{\mathbf{x}\in\mathcal{X}} P(\mathbf{y}|\mathbf{x})=\arg\min_{\mathbf{x}\in\mathcal{X}} \sum_{i=1}^n (y_i-x_i)^2
	\label{eq:MLE}
\end{equation}
which is also the \ac{ls} decoder.

Denoting by $\mathcal{M}$ the single-symbol constellation, we consider a uniform spherical codebook $\mathcal{X}\subset\mathcal{M}^n$, in which every codeword is of the same energy $E=nE\subscript{s}$ and equally likely to be sent (with $E\subscript{s}$ the energy per symbol), and a bounded decoder that allows erasures. While it is impractical to visualize the $n$-dimensional codeword space, we can still follow Shannon's approach to draw the distance-mapped hyper-spheres as shown in Fig.~\ref{fig:codebook_sphere}, where all codewords are located on the surface of the a sphere with radius $\sqrt{E}$. We denote the shortest Euclidean distance between two neighbor codewords as $D\subscript{min}$.
\begin{figure}[!htpb]
	\centering
	\includegraphics[width=.8\linewidth,clip,trim=100 50 100 60]{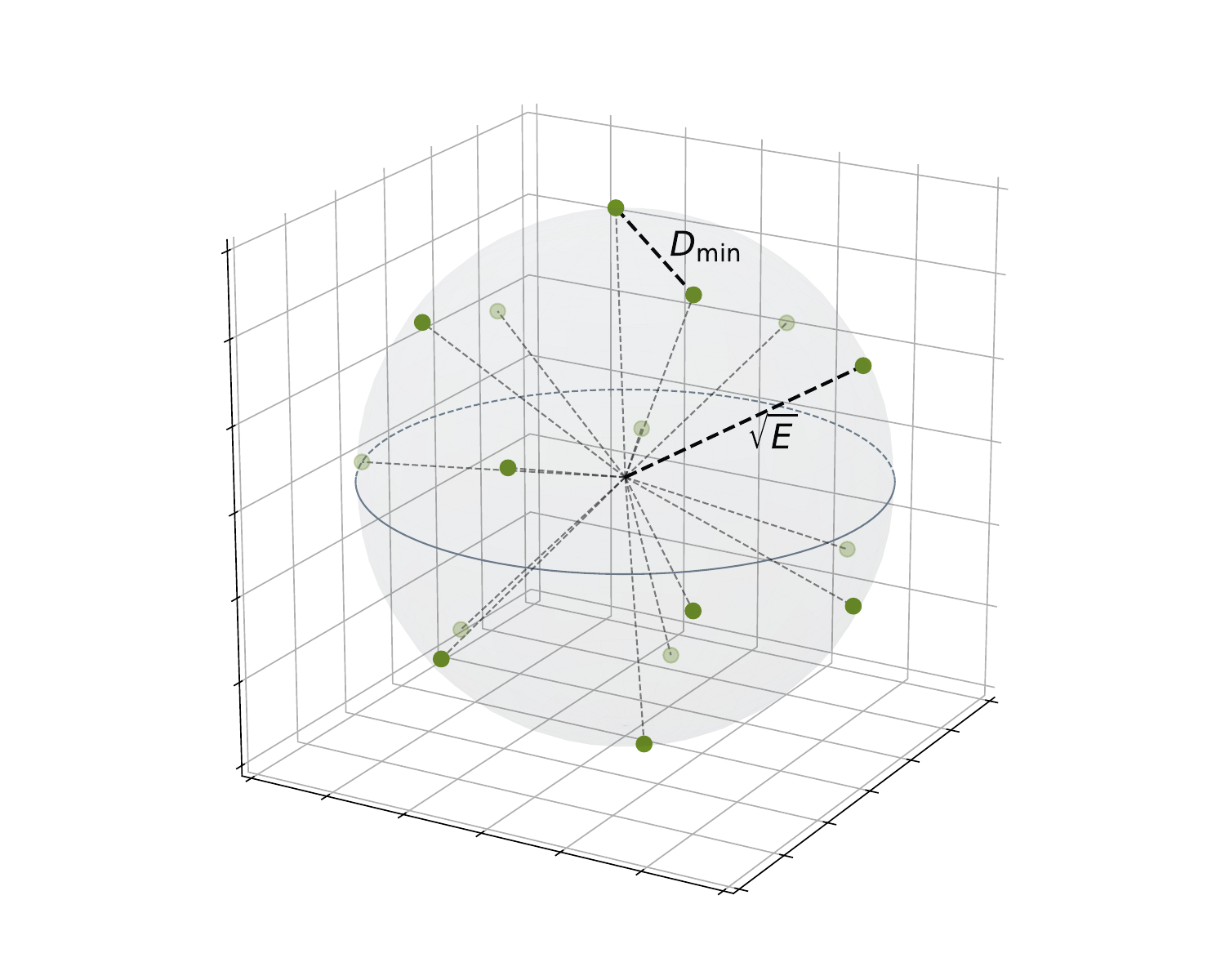}
	\caption{\ac{fbl} codewords projected onto a hyper-sphere}
	\label{fig:codebook_sphere}
\end{figure}

Now for an arbitrary codeword $\mathbf{x}\in\mathcal{X}$ be sent over the noisy channel, we denote the set of all its neighboring codewords in $\mathcal{X}$ with Hamming distance $\kappa$ as $\mathcal{V}(\mathcal{X},\mathbf{x},\kappa)$:
\begin{equation}
	\mathcal{V}(\mathcal{X},\mathbf{x},\kappa)\triangleq\left\{\mathbf{x}'\in\mathcal{X}\vert d(\mathbf{x},\mathbf{x}')=\kappa\right\},
\end{equation}
and therewith the total set of codewords in $\mathcal{X}$ other than $\mathbf{x}$:
\begin{equation}
	\mathcal{V}_\Sigma(\mathcal{X},\mathbf{x})=\underset{\kappa>0}{\cup}\mathcal{V}(\mathcal{X},\mathbf{x},\kappa).
\end{equation}
Consider the hyper-sphere centered at $\mathbf{x}$ and with another codeword $\mathbf{x}'\in\mathcal{V}_\Sigma(\mathcal{X},\mathbf{x})$ on its surface, as shown in Fig.~\ref{fig:decision_spheres}. The charcoal point and the cadet blue sphere around it represent $\mathbf{x}$ and its decision region at the decoder, respectively. Similarly, the olive drab point and the copper sphere around it represent $\mathbf{x'}$ and its decision region, respectively. Here we assume that all codewords have the same decision region radius $R$. For decoders with erasure but declining list-output, it must hold that $R\leqslant D\subscript{min}/2$. When $\mathbf{y}$ falls into the decision region of $\mathbf{x}$, the decoding is correct. When it falls into the decision region of any $\mathbf{x}'\in\mathcal{V}_\Sigma(\mathcal{X},\mathbf{x})$, a block confusion occurs. When it falls outside any of the decision regions, a block erasure occurs. 
\begin{figure}[!htpb]
	\centering
	\includegraphics[width=.8\linewidth,clip,trim=100 50 100 60]{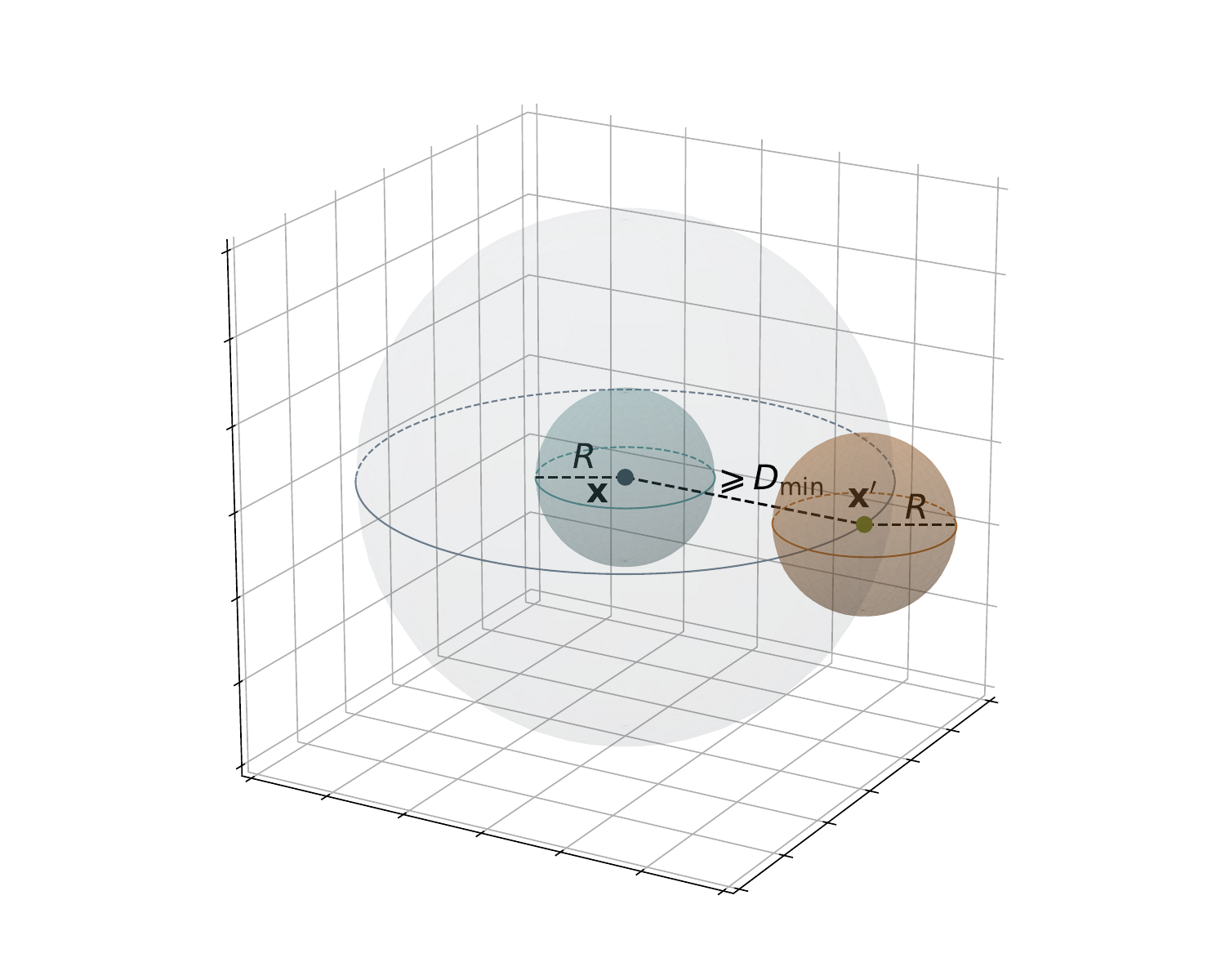}
	\caption{Decision spheres of two distinct codewords}
	\label{fig:decision_spheres}
\end{figure}

Thus, under the \ac{awgn} $\mathbf{w}$, the probability of block errors (including both block confusions and block erasures) is
\begin{equation}
	\varepsilon=P(\Vert\mathbf{w}\Vert\geqslant R)=\int_{R}^{+\infty} f_{\Vert\mathbf{w}\Vert}(w)\diff w=1-\int_{-\infty}^R f_{\Vert\mathbf{w}\Vert}(w)\diff w
	\label{eq:prob_incorrect}
\end{equation}

So for a $\varepsilon$-bounded decoder, i.e., one that only accepts the decoding result $\hat{\mathbf{x}}$ if the estimated block error rate is less than $\varepsilon$, we have the decision region radius as
\begin{equation}\label{eq:decision_region_radius}
	R(\varepsilon)=F^{-1}_{\Vert\mathbf{w}\Vert}(1-\varepsilon)=\sigma F^{-1}_{\chi_n}(1-\varepsilon),
\end{equation}
where $F_{\chi_n}(x)=P(n/2;x^2/2)$ is the \ac{cdf} of $\chi$ distribution with $n$ degrees of freedom, 
\begin{equation}
	P(s,x)=\frac{\gamma(s,x)}{\Gamma(s)}=\frac{\int_0^x t^{s-1}e^{-t}\diff t}{\int_0^\infty t^{s-1}e^{-t}\diff t}
\end{equation}
is the regularized gamma function.

Meanwhile, the probability of block confusion mistakening $\mathbf{x}$ into $\mathbf{x}'$ is given by
\begin{equation}
	\begin{split}
		&P(\mathbf{x}\to\mathbf{x'})\triangleq P(\hat{\mathbf{x}}=\mathbf{x}'|\mathbf{x})
		=P(\Vert\mathbf{w}-\mathbf{x'}+\mathbf{x}\Vert\leqslant R)
	\end{split}.\label{eq:undetected_err}
\end{equation}
To obtain this probability, we consider the conditional probability of $\mathbf{w}$ falls into the ball $B_R(\mathbf{x}'-\mathbf{x})$, given the noise norm $\Vert\mathbf{w}\Vert=w$, where $B_R(\mathbf{x})$ is the $R$-radius ball centered at $\mathbf{x}$. This conditional probability is captured by the spherical cap area of the ball $B_w(\mathbf{0})$ intersected with $B_R(\mathbf{x}'-\mathbf{x})$, divided by the total surface of $B_w(\mathbf{0})$:
\begin{equation}
	P\left[\mathbf{w}\in B_R(\mathbf{x}'-\mathbf{x})~\vert~\Vert\mathbf{w}\Vert=w\right]=\Omega_n[\theta_w(\Vert\mathbf{x}'-\mathbf{x}\Vert,R)]
\end{equation}
where $\theta_w$ is the cap angle and $\Omega_n$ the angle fraction:
\begin{align}
	\theta_w(D,R)&=\arccos\left(\frac{r^2+D^2-R^2}{2rD}\right)\\
	\Omega_n(\theta)&=\frac{\int_0^\theta\left(\sin\phi\right)^2\diff\phi}{\int_0^\pi\left(\sin\phi\right)^2\diff\phi}.\label{eq:angle_fraction}
\end{align}
Furthermore, we know that $\Vert\mathbf{w}\Vert/\sigma\sim\chi_n$, so that the pairwise block confusion rate~\eqref{eq:undetected_err} is dependent not on the direction of $\mathbf{x}'-\mathbf{x}$ but only its norm: 
\begin{equation}\label{eq:pairwise_con_prob}
	P(\mathbf{x}\to\mathbf{x}')=P\subscript{pair}(\Vert\mathbf{x}'-\mathbf{x}\Vert) 
\end{equation}
where
\begin{equation}\label{eq:con_prob_upon_distance}
	P\subscript{pair}(D)=\int\limits_{(D-R)/\sigma}^{(D+R)/\sigma}f_{\chi_n}(u)\Omega_n[\theta_w(\sigma u,R)]\diff u.
\end{equation}
The overall block confusion rate is
\begin{equation}\label{eq:overall_con_prob}
	\begin{split}
		P\subscript{con}=&\sum\limits_{\mathbf{x}\in\mathcal{X}}\left[P(\mathbf{x})\sum\limits_{\mathbf{x}'\in\mathcal{V}_\Sigma(\mathcal{X},\mathbf{x})}P(\mathbf{x}\to\mathbf{x}')\right]
		\\=&\frac{1}{\vert\mathcal{X}\vert}\sum\limits_{\substack{\mathbf{x}\in\mathcal{X}\\\mathbf{x}'\in\mathcal{V}_\Sigma(\mathcal{X},\mathbf{x})}}P(\mathbf{x}\to\mathbf{x}')
	\end{split}
\end{equation}
and the block erasure rate is given by
	$P\subscript{ers}=\varepsilon-P\subscript{con}$.
Clealry, each lower bound of $P\subscript{con}$ gives an upper bound of $P\subscript{ers}$, and each upper bound of $P\subscript{con}$ gives a lower bound of $P\subscript{ers}$. 

For generic \ac{ml} decoders, the sensitivity of the \ac{bler} and block erasure rates to $R$ is easy to capture:
\begin{theorem}\label{th:monotonicity_ers_and_con_regarding_R}
	Given fixed $\mathcal{X}$ with certain $(\mathcal{X}, E,\sigma^2)$, both $\varepsilon$ and $P\subscript{ers}$ are monotonically decreasing w.r.t. $R$, while $P\subscript{con}$ monotonically increases. 
\end{theorem}
\begin{proof}
	Trivial, omitted.
\end{proof}

However, in this study our interest mainly focuses on the coding aspect of this problem: given certain $(\varepsilon,E,\sigma^2)$ (so that with Eq.~\eqref{eq:decision_region_radius} $R$ is fixed), by selecting a proper codebook $\mathcal{X}\subseteq\mathcal{M}^n$, what is the achievably minimal $P\subscript{ers}$? For convenience of discussion, here we consider the case $\vert\mathcal{X}\vert=M^k$, i.e. the payload code space is fully utilized.

\section{Bound Analysis}\label{sec:bound_analysis}
\subsection{Bounds of the Minimum Distance between Codewords}
\label{subsec:bounds_dmin}
Eq.~\eqref{eq:pairwise_con_prob}--\eqref{eq:overall_con_prob} reveal that the block confusion probability is dominated by the Euclidean distance between distinct codewords. Moreover, given a tuple $(E,n)$, the Euclidean distance between any pair of codewords $(\mathbf{x},\mathbf{x}')\in\mathcal{M}^n\times\mathcal{M}^n$ is uniquely coupled with their Hamming distance $d(\mathbf{x},\mathbf{x}')$:
\begin{equation}\label{eq:hamming_to_euclidean}
	D(\mathbf{x},\mathbf{x}')=\sqrt{E\cdot d(\mathbf{x},\mathbf{x}')}
\end{equation}

Thus, we first focus on the Hamming distance, especially the mininum Hamming distance $d\subscript{min}(\mathcal{X})$ between distinct codewords, and can obtain the following lemma:
\begin{lemma}\label{lem:dmin_bounds}
	Given fixed $(M,n,k,E,\sigma)$, for any $\varepsilon$-bounded decoder that rejects list-output, the minimum Hamming distance $d\subscript{min}(\mathcal{X})$ between any pair of distinct  codewords in a $M^k$-sized codebook $\mathcal{X}$ is always bounded between $[d\subscript{min}\superscript{min}, d\subscript{min}\superscript{max}]$ where
	\begin{align}
		d\subscript{min}\superscript{min}&=\left\lceil\frac{4R^2(\varepsilon)}{E}\right\rceil\label{eq:lower_bound_dmin},\\
		d\subscript{min}\superscript{max}&=2\min\left\{t\in\mathbb{N}\left\vert\sum\limits_{l=0}^t\comb{n}{l}(M-1)^l>M^r\right.\right\}.\label{eq:upper_bound_dmin}
	\end{align}
\end{lemma}
\begin{proof}
	See Appendix~\ref{app:dmin_bounds}.
\end{proof}
These bounds can be transformed with Eq.~\eqref{eq:hamming_to_euclidean} into the corresponding Euclidean distance bounds:
\begin{equation}
	\sqrt{Ed\subscript{min}\superscript{min}}\leqslant D\subscript{min}(\mathcal{X}) \leqslant \sqrt{Ed\subscript{min}\superscript{max}}.
\end{equation}
\subsection{Lower Bound of the Block Confusion Rate}
Subsequently, we notice the monotonicity of $P\subscript{pair}$ w.r.t. $D$:
\begin{lemma}\label{lem:monotonicity_Ppair}
	Given fixed $(E,\sigma)$, the pairwise block confusion rate $P\subscript{pair}$ is a monotonically decreasing and convex function of $D=\sqrt{Ed}$ for all $D\geqslant 2R(\varepsilon)$.
\end{lemma}
\begin{proof}
	See Appendix~\ref{app:monotonicity_Ppair}.
\end{proof}
This reveals a lower bound of the confusion rate $P\subscript{con}$:
\begin{theorem}
	\label{th:lower_bound_Pcon}
	Given fixed $(M,n,k,E)$, for any $\varepsilon$-bounded decoder that rejects list-output, the block confusion rate $P\subscript{con}$ is lower-bounded by
	\begin{equation}\label{eq:lower_bound_Pcon}
		P\subscript{con}\superscript{LB}\triangleq\comb{n}{d\subscript{min}\superscript{max}}\frac{(M-1)^{d\subscript{min}\superscript{max}}}{M^{n-k}}P\subscript{pair}\left(\sqrt{Ed\subscript{min}\superscript{max}}\right).
	\end{equation}
\end{theorem}
\begin{proof}
	See Appendix~\ref{app:lower_bound_Pcon}.
\end{proof}
Correspondingly, the block erasure rate is upper-bounded:
\begin{equation}
	P\subscript{ers}\leqslant\varepsilon-P\subscript{con}\superscript{LB}.
	\label{eq:upper_bound_Pers}
\end{equation}

\subsection{Upper Bound of the Block Confusion Rate}
On the other hand, an upper bound of the block confusion rate can be obtained:
\begin{theorem}\label{th:upper_bound_Pcon}
	Given fixed $(M,k,E)$, for any $\varepsilon$-bounded decoder that rejects list-output, the block confusion rate $P\subscript{con}$ is upper-bounded by
	\begin{equation}\label{eq:upper_bound_Pcon}
		P\subscript{con}\superscript{UB}\triangleq (M^k-1)P\subscript{pair}\left(\sqrt{Ed\subscript{min}\superscript{min}}\right).
	\end{equation}
\end{theorem}
\begin{proof}
	See Appendix~\ref{app:upper_bound_Pcon}.
\end{proof}
Again, it lower-bounds the block erasure rate with:
\begin{equation}\label{eq:lower_bound_Pers}
	P\subscript{ers}\geqslant \varepsilon-P\subscript{con}\superscript{UB}.
\end{equation}

\section{Sensitivity of the Confusion Rate Bounds}\label{sec:sensitivity_analysis}
Having derived the error bounds, we are interested in their sensitivity to the system parameters. Since $P\subscript{ers}$ and $P\subscript{con}$ are one-to-one coupled under fixed $\varepsilon$, we focus here on the bounds $P\subscript{con}\superscript{LB}$ and  $P\subscript{con}\superscript{UB}$.
\subsection{Lower Bound versus Power}
First, we analyze the impact of $E$ on $P\subscript{con}\superscript{LB}$ and can derive:
\begin{theorem}
	\label{th:monotonicity_Pcon_LB_regarding_E}
	Given fixed $(M, k, n, \sigma, \varepsilon)$, the lower bound $P\subscript{con}\superscript{LB}$ is a monotonically decreasing and convex function of $E$.
\end{theorem}
\begin{proof}
	See Appendix~\ref{app:monotonicity_Pcon_LB_regarding_E}.
\end{proof}

\subsection{Lower Bound versus Blocklength}
When investigating the sensitivity of $P\subscript{con}\superscript{LB}$ regarding the blocklength $n$, we must take account that $d\subscript{min}\superscript{max}$ is a function of $n$ as defined in Eq.~\eqref{eq:upper_bound_dmin}, and shall be analyzed first.
\begin{lemma}
	\label{lem:monotonicity_dmin_max_regarding_n}
	Given fixed $(M, k)$, the maximum minimum Hamming distance $d\subscript{min}\superscript{max}$ is a monotonically increasing (discrete) function of $n$.
\end{lemma}
\begin{proof}
	See Appendix~\ref{app:monotonicity_dmin_max_regarding_n}.
\end{proof}

\begin{lemma}
	\label{lem:bound_dmin_max_regarding_n}
	$\forall n\in[k, M^k-1]$, $\exists~\lambda<\frac{2(M-1)}{M}+\frac{1}{n}$ such that $d\subscript{min}\superscript{max}(n)\leqslant\lambda n$
\end{lemma}
\begin{proof}
	See Appendix~\ref{app:bound_dmin_max_regarding_n}.
\end{proof}

\begin{lemma}\label{lem:bound_lambda}
	As long as $k>\frac{2M}{\log_M(2M^2)}-1$, the bound $\lambda$ in Lemma~\ref{lem:bound_dmin_max_regarding_n} always fulfill $\lambda<\frac{1}{M}$.
\end{lemma}
\begin{proof}
	See Appendix~\ref{app:bound_lambda}.
\end{proof}

\begin{theorem}
	\label{th:monotonicity_PconLB_regarding_n}
	Given fixed $(M, k, E, \sigma, \varepsilon)$, the lower bound $P\subscript{con}\superscript{LB}$ is piecewise monotonic that decreases with increasing $n$ in every individual $d\subscript{min}\superscript{max}$-consistent interval as long as $k>\frac{2M}{\log_M(2M^2)}-1$. For large enough $n$, it is monotonic.
\end{theorem}
\begin{proof}
	See Appendix~\ref{app:monotonicity_PconLB_regarding_n}.
\end{proof}

\subsection{Upper Bound versus Power}
First we consider the behavior of $d\subscript{min}\superscript{min}$ , which strongly impacts the block confusion rate's upper bound $P\subscript{con}\superscript{max}$:
\begin{lemma}
	For any fixed $(\sigma, n, \varepsilon)$, $d\subscript{min}\superscript{min}$ is wide-sense monotonically decreasing w.r.t. $E$.
\end{lemma}
\begin{proof}
As defined by Eq.~\eqref{eq:lower_bound_dmin}:
\begin{equation}
	d\subscript{min}\superscript{min}=\left\lceil\frac{4R^2(\varepsilon)}{E}\right\rceil=\left\lceil\frac{4\sigma^2[F^{-1}_{\chi_n}(1-\varepsilon)]^2}{E}\right\rceil.
\end{equation}	
Since $F^{-1}_{\chi_n}(1-\varepsilon)$ is a constant for fixed $(\sigma, n, \varepsilon)$, the wide-sense monotonicity is obvious.
\end{proof}

Based on this, we can further derive the following features of the upper bound $P\subscript{con}\superscript{UB}$ of block confusion rate:
\begin{theorem}
	\label{th:Pcon_UB_regarding_E}
	Given fixed $(M, k, d\subscript{min}\superscript{min})$, $P\subscript{con}\superscript{UB}$ is piecewise continuous regarding $E$, with jump discontinuities at $E_i=4R^2(\varepsilon)/i$ where $i\in\mathbb{N}^+$, and decreases monotonically in each individual continuous interval. 
\end{theorem}
\begin{proof}
	See Appendix~\ref{app:Pcon_UB_regarding_E}.
\end{proof}
\begin{corollary}
	\label{co:Pcon_UB_regarding_E_extremes}
	Given fixed $(M, k, d\subscript{min}\superscript{min})$, $P\subscript{con}\superscript{UB}$ has its local maximums regarding $E$ at points $E_i=4R^2(\varepsilon)/i$ where $i\in\mathbb{N}^+$, and its local minimums at $E_i^-=\lim\limits_{\delta\to 0}4R^2(\varepsilon)/i-\vert\delta\vert$. All local maximums are equal to each other, and the local minimums are monotonically decreasing w.r.t. $i$.
\end{corollary}
\begin{proof}
	See Appendix~\ref{app:Pcon_UB_regarding_E_extremes}.
\end{proof}

\begin{figure}[!htpb]
	\centering
	\begin{subfigure}{\linewidth}
		\centering
		\includegraphics[width=.95\linewidth, trim=5 15 5 5, clip]{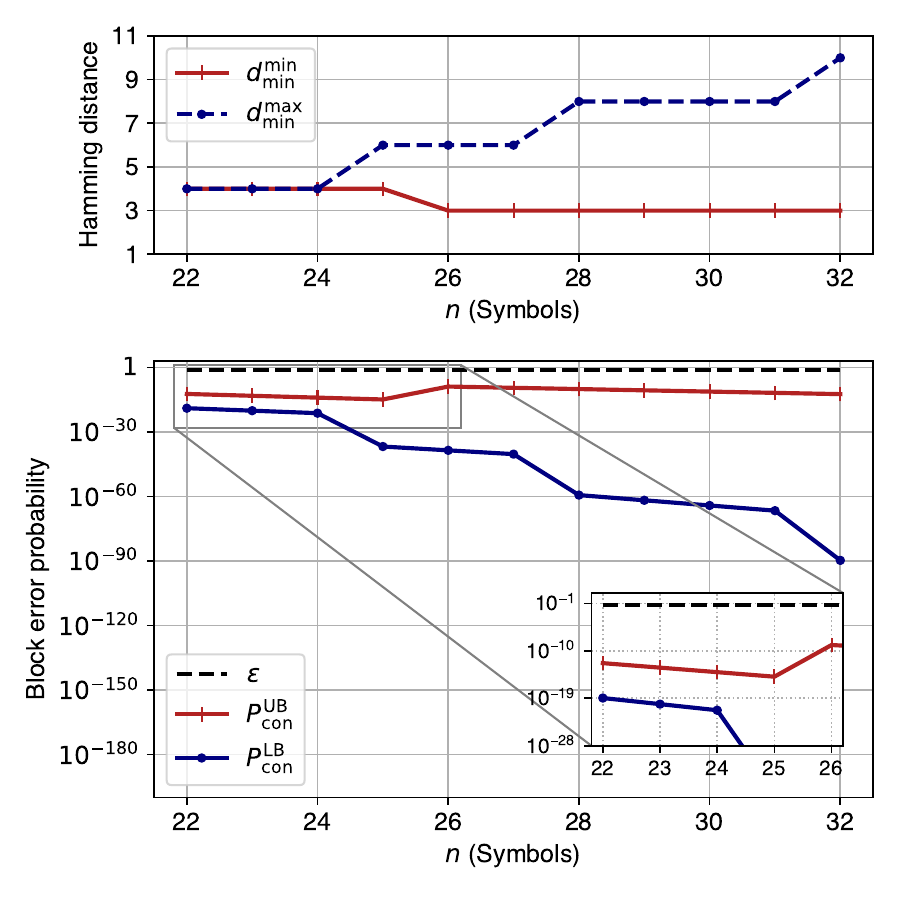}
		\subcaption{$k=16$}
		\label{subfig:p2_bounds_regarding_n_s1}
	\end{subfigure}\\
	\begin{subfigure}{\linewidth}
		\centering
		\includegraphics[width=.95\linewidth, trim=5 15 5 5, clip]{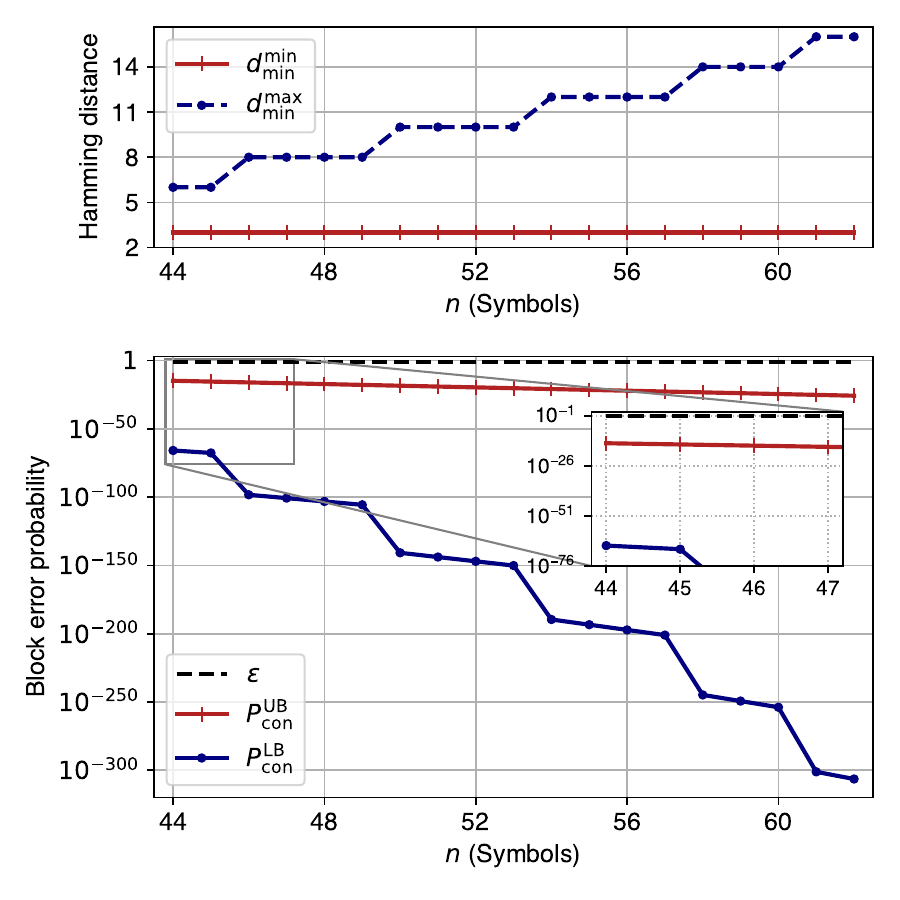}
		\subcaption{$k=32$}
		\label{subfig:p2_bounds_regarding_n_s2}
	\end{subfigure}
	\caption{Block confusion rate bounds vs. $n$ with $E\subscript{b}/N_0=\SI{0}{\dB}$.}
	\label{fig:p2_bounds_regarding_n}
\end{figure}

\subsection{Upper Bound versus Blocklength}
 On the other hand, since $F^{-1}_{\chi_n}(1-\varepsilon)$ strictly decreases with increasing $n$, we know that for any fixed $(\sigma, E, \varepsilon)$, $d\subscript{min}\superscript{min}$ is \emph{wide-sense monotonically decreasing} w.r.t. $n$ as well, which ensures with Eq.~\eqref{eq:upper_bound_Pcon} that
\begin{theorem}\label{th:Pcon_UB_regarding_n}
	Given fixed $(M, k, E, \sigma, \varepsilon)$, $P\subscript{con}\superscript{UB}$ is a piecewise monotonic (discrete) function decreasing in each individual $d\subscript{min}\superscript{min}$-consistent interval.
\end{theorem}
\begin{proof}
	Trivial.
\end{proof}
However, at the boundaries between such $d\subscript{min}\superscript{min}$-consistent intervals, i.e. when $d\subscript{min}\superscript{min}(n+1)=d\subscript{min}\superscript{min}(n)-1$, since $P\subscript{pair}^{(n)}(D)$ (see definition in Appendix \ref{app:monotonicity_PconLB_regarding_n}) decreases exponentially w.r.t. $n$, and meanwhile decreases hyper-exponentially w.r.t. $D$, the inequality between $P\subscript{pair}^{(n)}\left(\sqrt{E d\subscript{min}\superscript{min}(n)}\right)$ and $P\subscript{pair}^{(n+1)}\left(\sqrt{E d\subscript{min}\superscript{min}(n)-1}\right)$ is complex to track, and the global monotonicity of $P\subscript{con}\superscript{UB}$ regarding $n$ is therefore not guaranteed.


\section{Numerical Results}\label{sec:results}
To verify the analytical results presented in Sec.~\ref{sec:sensitivity_analysis}, we numerically computed the bounds of minimum Hamming distance and block confusion rate under selected system setups. More specifically, we set $M=2,~\varepsilon=0.05,~\sigma^2=0.5$, and tested the bounds under different payload lengths, blocklengths, and signal-to-noise ratios.



Fig.~\ref{fig:p2_bounds_regarding_n} shows how the block confusion rate bounds against blocklength $n$ for payload lengths $k = 16$ and $k = 32$. As $n$ increases, both the lower and upper bounds of the confusion rate drop rapidly, demonstrating the (piecewise) monotonic improvement of decoder robustness with block length, which is aligned with Theorems~\ref{th:monotonicity_PconLB_regarding_n} and \ref{th:Pcon_UB_regarding_n}. 
However, as we discussed in~\ref{th:Pcon_UB_regarding_n}, the monotonicity does not hold globally due to competing effects in $P\subscript{pair}^{(n)}(D)$.  
Notably, both confusion rate bounds remain many orders of magnitude below the decoding error probability with $\varepsilon = 0.05$. It confirms that almost all block errors in this regime are erasures rather than confusions. 

\begin{figure}[!htpb]
	\centering
	\begin{subfigure}{\linewidth}
		\centering
		\includegraphics[width=.9\linewidth, trim=5 15 5 5, clip]{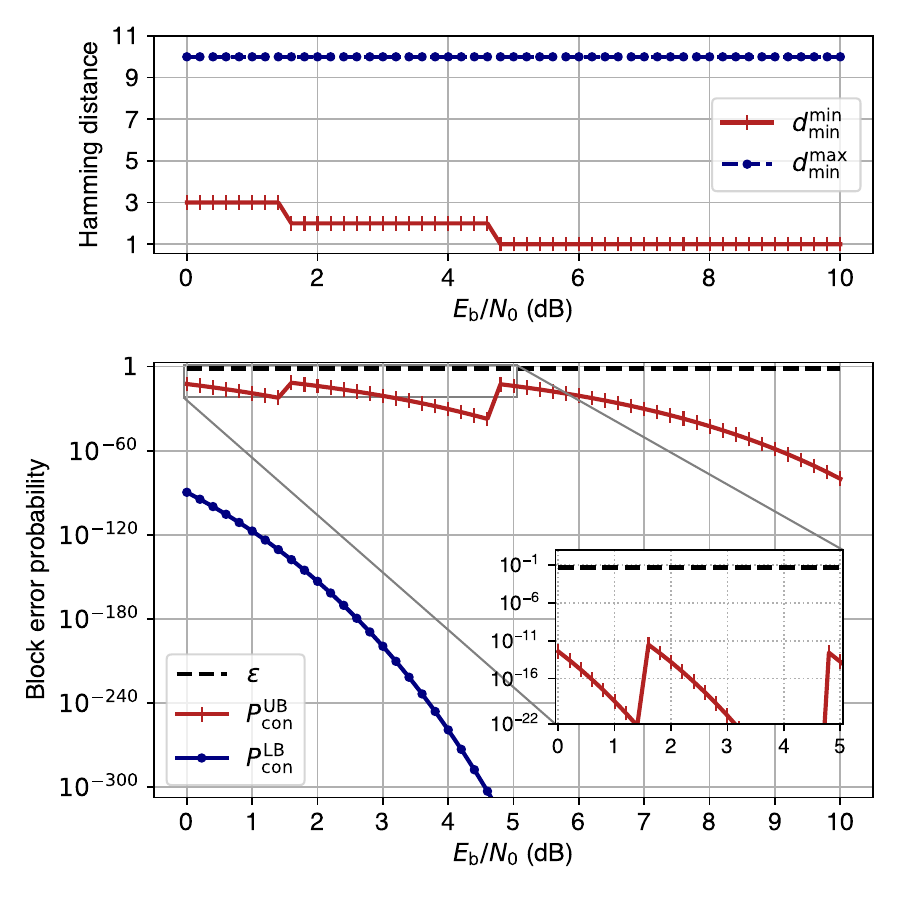}
		\subcaption{With $(32,16)$ codebooks}
		\label{subfig:p2_bounds_regarding_ebn0_s1}
	\end{subfigure}\\
	\begin{subfigure}{\linewidth}
		\centering
		\includegraphics[width=.9\linewidth, trim=5 15 5 5, clip]{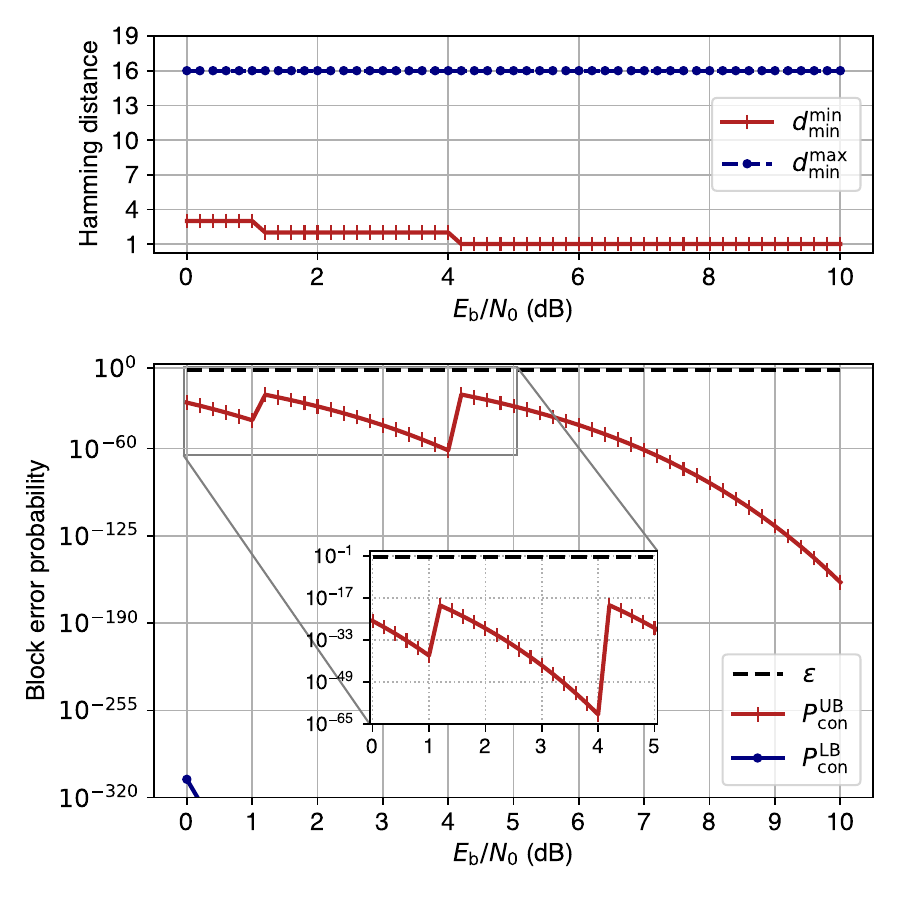}
		\subcaption{With $(62,32)$ codebooks}
		\label{subfig:p2_bounds_regarding_ebn0_s2}
	\end{subfigure}
	\caption{Block confusion rate bounds vs. $E\subscript{b}/N_0$. Note the $P\subscript{con}\superscript{LB}$ is convex as proven in Theorem~\ref{th:monotonicity_Pcon_LB_regarding_E}, but may look concave here due to the logarithmic $y$-scale.}
	\label{fig:p2_bounds_regarding_ebn0}
\end{figure}

Figure~\ref{fig:p2_bounds_regarding_ebn0} plots the block confusion rate bounds versus $E_b/N_0$ for two codebook setups, i.e., $(32, 16)$ and $(62, 32)$. It shows that both the lower and upper bounds decrease with increasing SNR, in agreement with Theorem~\ref{th:monotonicity_ers_and_con_regarding_R}. We can also observe that the step-wise increases in $d_{\min}^{\min}$ and $d_{\min}^{\max}$ induce discontinuities in the upper bound curves. This is due to the discrete nature of distance bounds under fixed decoding radius. Similar to Fig.~\ref{fig:p2_bounds_regarding_n}, the confusion bounds remain several orders below the BLER constraint, again verifying that confusion events are practically negligible compared to erasures.

\section{Conclusion and Outlooks}\label{sec:conclusion}
In this paper, we have derived and analyzed the probability bounds of confusions and erasures of error-bounded block decoders with finite blocklength, including the sensitivity of these bounds to the blocklength and \ac{snr}. Numerical results have confirmed these findings, and revealed that the block confusion rate is extremely low compared to the block erasure rate, supporting to approximately consider all block errors in the \ac{fbl} regime as erasures.

For future work, we see potential interest in extending the analysis to more general codebooks beyond the spherical uniform assumption. Specifically, spherical codebooks require either \ac{psk} constellations or permutation codes, which present practical construction challenges for general applications. Additionally, non-uniform codeword distributions may offer efficiency advantages over the uniform distribution assumed in this work.


\appendices

\section{Proof of Lemma~\ref{lem:dmin_bounds}}\label{app:dmin_bounds}
\begin{proof}
	Given $E,\sigma,\varepsilon$, the minimum Euclidian distance between the closest pair of codewords in $\mathcal{X}$ that rejects list-output decoding is simply determined by
	\begin{equation}
		D\subscript{min}(\mathcal{X})\geqslant 2R(\varepsilon)=2\sigma F^{-1}_{\chi_n}(1-\varepsilon),
		\label{eq:lower_bound_Dmin}
	\end{equation}
	and therefore the \emph{integer} minimum Hamming distance
	\begin{equation}
		d\subscript{min}(\mathcal{X})\geqslant \left\lceil\frac{4R^2(\varepsilon)}{E}\right\rceil.
	\end{equation}

	To derive the upper bound $d\subscript{min}\superscript{max}$, we leverage the Hamming bound: the number of $M$-ary codeswords in a codebook $\mathcal{X}$ with minimum Hamming distance $d\subscript{min}$ is bounded by
	\begin{equation}
		A_M(n,d\subscript{min})\leqslant \frac{M^n}{\sum\limits_{l=0}^t\comb{n}{l}(M-1)^l}
		\label{eq:hamming_bound}	
	\end{equation}
	where $t = \lfloor (d\subscript{min}-1)/2\rfloor$, so that $2t+1\leqslant d\subscript{min}<2t+3$. Since $d\subscript{min}\in\mathbb{N}^+$, we can further tighten it to $d\subscript{min}\leqslant 2t+2$.
	Especially, since we are considering $\vert\mathcal{X}\vert=M^k$, we can replace $A_M(n,d\subscript{min})$ with $M^k$ so that
	\begin{equation}
		\sum_{l=0}^t\comb{n}{l}(M-1)^l \leqslant M^{n-k} = M^r
	\end{equation}
	and thererfore
	\begin{equation}
		d\subscript{min}(\mathcal{X})\leqslant 2\min\left\{t\in\mathbb{N}\left\vert\sum\limits_{l=0}^t\comb{n}{l}(M-1)^l>M^r\right.\right\}
	\end{equation}
\end{proof}

\section{Proof of Lemma~\ref{lem:monotonicity_Ppair}}\label{app:monotonicity_Ppair}
\begin{proof}

	First, regarding the monoticity: we recall the angle fraction given by Eq.~\eqref{eq:angle_fraction} that $\Omega_n(\theta) = C_n \int_0^\theta \left(\sin\phi\right)^2 \diff\phi$ where $C_n=\left(\int_0^\pi\left(\sin\phi\right)^{n-2}\diff\phi\right)^{-1}$. Its first and second derivatives are
	\begin{align}
		\frac{\partial\Omega_n}{\partial\theta}&= C_n \left(\sin\theta\right)^{n-2},\text{~and}\\
		\frac{\partial^2 \Omega_n}{\partial\theta^2} &= (n-2)C_n\left(\sin\theta\right)^{n-3} \cos\theta,	
	\end{align}
	respectively. Noting the boundary terms vanish as $\Omega_n(0)=0$:
	\begin{equation}
	\frac{\partial P\subscript{pair}}{\partial D} = \int_{u_-}^{u_+} f_{\chi_n}(u) \frac{\partial\Omega_n}{\partial\theta_w} \frac{\partial \theta_w}{\partial D}
	\diff u
	\end{equation}
	From $\cos\theta_w = \frac{r^2 + D^2 - R^2}{2rD}$, we obtain $\frac{\partial \theta_w}{\partial D} = - \frac{D^2 - r^2 + R^2}{2rD^2 \sin \theta_w} < 0$ for $r \in (\sigma u_-, \sigma u_+)$. Thus, $\frac{\partial P\subscript{pair}}{\partial D}<0$, the monoticity is proven.

	Then, regarding the convexity:
	\begin{equation}
		\begin{split}
			\frac{\partial^2P\subscript{pair}}{\partial D^2} = \int_{u_-}^{u_+} f_{\chi_n}(u) \left[\frac{\partial^2\Omega_n}{\partial\theta_w^2} \left( \frac{\partial \theta_w}{\partial D} \right)^2 +\frac{\partial\Omega_n}{\partial\theta_w} \frac{\partial^2 \theta_w}{\partial D^2} \right]\diff u
		\end{split}
	\end{equation}
	For $D \geqslant 2R$, $\theta_w < \pi/2$, which ensures $\frac{\partial^2\Omega_n}{\partial\theta_w^2} > 0$. Furthermore, the geometric decay of the subtended angle $\theta_w$ with respect to $D$ is convex, i.e., $\frac{\partial^2 \theta_w}{\partial D^2} > 0$. Since $f_{\chi_n}(u) > 0$, it follows that $\frac{\partial^2 P\subscript{pair}}{\partial D^2}>0$, confirming the convexity.
\end{proof}

\section{Proof of Theorem~\ref{th:lower_bound_Pcon}}\label{app:lower_bound_Pcon}
\begin{proof}
	Recall Eq.~\eqref{eq:overall_con_prob} and leverage Lemma~\ref{lem:dmin_bounds}, we can obtain a lower bound of the block confusion rate by considering only the confusions between closest neighbor codeword pairs:
	\begin{equation}\label{eq:lower_bound_Pcon_preliminary}
		\begin{split}
			P\subscript{con}&=\frac{1}{\vert\mathcal{X}\vert}\sum\limits_{\substack{\mathbf{x}\in\mathcal{X}\\\mathbf{x}'\in\mathcal{V}_\Sigma(\mathcal{X},\mathbf{x})}}P(\mathbf{x}\to\mathbf{x}')
			\\&\geqslant \frac{1}{\vert\mathcal{X}\vert}\sum\limits_{\substack{\mathbf{x}\in\mathcal{X}\\\mathbf{x}'\in\mathcal{V}(\mathcal{X},\mathbf{x},d\subscript{min})}}P\subscript{pair}\left(D\subscript{min}(\mathcal{X})\right)
			\\&\geqslant \frac{P\subscript{pair}\left(\sqrt{Ed\subscript{min}\superscript{max}}\right)}{\vert\mathcal{X}\vert}\sum\limits_{\mathbf{x}\in\mathcal{X}}\left\vert\mathcal{V}(\mathcal{X},\mathbf{x},d\subscript{min}\superscript{max})\right\vert,
		\end{split}
	\end{equation}
	where 
	the equality holds if and only if all pair of distinct codewords share the same distance, \emph{which is only possible for $\Vert\mathcal{X}\Vert=4$}.
	It is trivial to see that $d\subscript{min}\superscript{max}$ is achieved when the $M^k$ codewords are possibly most uniformly distributed in the $n$-dimensional $M$-ary space $\mathcal{M}^n$, which allows to estimate
	\begin{equation}
		\begin{split}
			\left\vert\mathcal{V}\left(\mathcal{X},\mathbf{x},d\subscript{min}\superscript{max}\right)\right\vert&\approx\frac{M^k}{M^n}\left\vert\mathcal{V}\left(\mathcal{M}^n,\mathbf{x},d\subscript{min}\superscript{max}\right)\right\vert
			\\&=\frac{1}{M^{n-k}}\comb{n}{d\subscript{min}\superscript{max}}(M-1)^{d\subscript{min}\superscript{max}}
		\end{split}
	\end{equation}
	with high accuracy, and therewith turns Eq.~\eqref{eq:lower_bound_Pcon_preliminary} into 
	\begin{equation}
		P\subscript{con}\geqslant \comb{n}{d\subscript{min}\superscript{max}}\frac{(M-1)^{d\subscript{min}\superscript{max}}}{M^{n-k}}P\subscript{pair}\left(\sqrt{Ed\subscript{min}\superscript{max}}\right).
	\end{equation}
\end{proof}

\section{Proof of Theorem~\ref{th:upper_bound_Pcon}}\label{app:upper_bound_Pcon}
\begin{proof}
	We leverage the $D$-monotonicity of $P\subscript{pair}$ (Lemma~\ref{lem:monotonicity_Ppair}) that for all $(\mathbf{x},\mathbf{x}')\in\mathcal{X}^2$,
	$
		P(\mathbf{x}\to\mathbf{x}')\leqslant P\subscript{pair}\left(D\subscript{min}(\mathcal{X})\right)
	$,
	so
	\begin{equation}
		\begin{split}
			P\subscript{con}&=\frac{1}{\vert\mathcal{X}\vert}\sum\limits_{\substack{\mathbf{x}\in\mathcal{X}\\\mathbf{x}'\in\mathcal{V}_\Sigma(\mathcal{X},\mathbf{x})}}P(\mathbf{x}\to\mathbf{x}')
			\\&=\frac{1}{\vert\mathcal{X}\vert}\sum\limits_{\substack{\mathbf{x}\in\mathcal{X}\\\mathbf{x}'\in\mathcal{V}_\Sigma(\mathcal{X},\mathbf{x})}}P\subscript{pair}(\Vert\mathbf{x}-\mathbf{x}'\Vert)
			\\&\leqslant\frac{1}{\vert\mathcal{X}\vert}\sum\limits_{\substack{\mathbf{x}\in\mathcal{X}\\\mathbf{x}'\in\mathcal{V}_\Sigma(\mathcal{X},\mathbf{x})}}P\subscript{pair}(D\subscript{min}(\mathcal{X}))
			\\&\leqslant\frac{P\subscript{pair}\left(\sqrt{Ed\subscript{min}\superscript{min}}\right)}{\vert\mathcal{X}\vert}\sum\limits_{\mathbf{x}\in\mathcal{X}}\left\vert\mathcal{V}_\Sigma(\mathcal{X},\mathbf{x})\right\vert
			\\&\leqslant P\subscript{con}\superscript{UB},
		\end{split}
	\end{equation}
	where the equality holds when $\Vert\mathbf{x}-\mathbf{x}'\Vert$ is the same for all $(\mathbf{x},\mathbf{x}')$, \emph{which is only possible for $\Vert\mathcal{X}\Vert=4$}.
\end{proof}

\section{Proof of Theorem~\ref{th:monotonicity_Pcon_LB_regarding_E}}\label{app:monotonicity_Pcon_LB_regarding_E}
\begin{proof}
	As defined in Eq.~\eqref{eq:upper_bound_dmin}, $d\subscript{min}\superscript{max}$ is independent from $E$ or $\varepsilon$ but solely determined by the tuple $(M, r, n)$. Thus, given certain $(M, r, n)$, $d\subscript{min}\superscript{max}$ is a fixed constant, which ensures the continuity and monotonicity of $P\subscript{con}\superscript{LB}$ regarding $E$:
	\begin{align}
		&\frac{\partial P\subscript{con}\superscript{LB}}{\partial E}=\comb{n}{d\subscript{min}\superscript{max}}\frac{(M-1)^{d\subscript{min}\superscript{max}}}{M^{n-k}}\frac{\partial}{\partial E}P\subscript{pair}\left(\sqrt{Ed\subscript{min}\superscript{max}}\right)\nonumber\\
		=&\underset{>0}{\underbrace{\comb{n}{d\subscript{min}\superscript{max}}\frac{(M-1)^{d\subscript{min}\superscript{max}}}{2 M^{n-k}}\sqrt{\frac{d\subscript{min}\superscript{max}}{E}}}}\underset{<0\text{ (Lemma~\ref{lem:monotonicity_Ppair})}}{\underbrace{
		\left.\frac{\partial P\subscript{pair}(D)}{\partial D}\right\vert_{D=\sqrt{Ed\subscript{min}\superscript{max}}}}}\nonumber\\
		<&0.
	\end{align}
	Moreover, for convenience of notation, we define the auxiliary variable
	$
		\Lambda\triangleq\comb{n}{d\subscript{min}\superscript{max}}\frac{(M-1)^{d\subscript{min}\superscript{max}}}{M^{n-k}},
	$
	and have
	\begin{equation}
		\begin{split}
		\frac{\partial^2 P\subscript{con}\superscript{LB}}{\partial E^2}=&\underset{<0}{\underbrace{-\frac{\Lambda}{4}\sqrt{\frac{d\subscript{min}\superscript{max}}{E^3}}}}\underset{<0\text{ (Lemma~\ref{lem:monotonicity_Ppair})}}{\underbrace{\frac{\partial P\subscript{pair}(D)}{\partial D}\bigg\vert_{D=\sqrt{Ed\subscript{min}\superscript{max}}}}}\\
		&+\underset{>0}{\underbrace{\frac{\Lambda}{4}\frac{d\subscript{min}\superscript{max}}{E}}}\underset{>0\text{ (Lemma~\ref{lem:monotonicity_Ppair})}}{\underbrace{\frac{\partial^2 P\subscript{pair}(D)}{\partial D^2}\bigg\vert_{D=\sqrt{Ed\subscript{min}\superscript{max}}}}}\\
		>&0,
		\end{split}
	\end{equation}
	which ensures the convexity of $P\subscript{con}\superscript{LB}$ regarding $E$.
\end{proof}

\section{Proof of Lemma~\ref{lem:monotonicity_dmin_max_regarding_n}}\label{app:monotonicity_dmin_max_regarding_n}
\begin{proof}
	For convenience of notation, we define auxiliaries
	$V_n(t)\triangleq\sum\limits_{l=0}^t\comb{n}{l}(M-1)^l$ and $T_n\triangleq d\subscript{min}\superscript{max}(n)/2$,	
	from Eq.~\eqref{eq:upper_bound_dmin} we know
	\begin{align}
		V_n(T_n-1)&\leqslant M^{n-k}\label{eq:Vn_Tn-1}
		,\\V_n(T_n)&> M^{n-k}
		,\\V_{n+1}(T_{n+1}-1)&\leqslant M^{n+1-k}
		,\\V_{n+1}(T_{n+1})&> M^{n+1-k}\label{eq:Vn+1_Tn+1}
		.
	\end{align}
	For all positive integers $(n,l,M)$ that $n>l$ and $M\geqslant 2$:
	\begin{equation}
		\begin{cases}
			\comb{n-1}{l}+\comb{n-1}{l-1}=\comb{n}{l} &\text{(Pascal's rule)}\\
			\comb{n}{0}(M-1)^0=1,&
		\end{cases}
	\end{equation}
	\begin{equation}
		\begin{split}\therefore
		&V_{n+1}(t)=\sum\limits_{l=0}^t\comb{n+1}{l}(M-1)^l\\
		=&1+\sum\limits_{l=1}^t\comb{n+1}{l}(M-1)^l\\
		=&1+\sum\limits_{l=1}^t\left[\comb{n}{l}+\comb{n}{l-1}\right](M-1)^l\\
		=&\sum\limits_{l=0}^t\comb{n}{l}(M-1)^l+\sum\limits_{q=0}^{t-1}\comb{n}{q}(M-1)^{q+1}\\
		=&V_n(t)+(M-1)V_n(t-1)\\
		<&MV_n(t).
		\end{split}	
	\end{equation}
	Thus, it always holds that
	\begin{equation}
		V_n(t)\leqslant M^{n-k}\Rightarrow V_{n+1}(t)\leqslant M^{n+1-k}.
	\end{equation}
	Linking this with Eqs.~\eqref{eq:Vn_Tn-1}--\eqref{eq:Vn+1_Tn+1}, we have
	\begin{equation}
		\forall t\in\mathbb{N}^+: t<T_n\Rightarrow t<T_{n+1},
	\end{equation}
	which implies that $T_{n+1}\geqslant T_n$, and therefore
	\begin{equation}
		d\subscript{min}\superscript{max}(n+1)\geqslant d\subscript{min}\superscript{max}(n).
	\end{equation}
\end{proof}

\section{Proof of Lemma~\ref{lem:bound_dmin_max_regarding_n}}\label{app:bound_dmin_max_regarding_n}
\begin{proof}
	Define $\delta =\left[d\subscript{min}\superscript{max}(n)-1\right]/2n$, obviously it always holds $\delta<1/2$. Meanwhile,from Eq.~\eqref{eq:upper_bound_dmin} we have
	\begin{equation}
	\sum_{l=0}^{\delta n} \binom{n}{l} (M-1)^l \leqslant M^{n-k}.
	\end{equation}
	With the $M$-ary entropy function
	\begin{equation}\label{eq:M-ary_entropy}
		H_M(\delta)=\delta\log_M(M-1)-\delta\log_M\delta-(1-\delta)\log_M(1-\delta),
	\end{equation}
	we have the wel-known bound 
	\begin{equation}
	\frac{1}{n+1}M^{n H_M(\delta)} \leqslant \sum_{l=0}^{\delta n} \binom{n}{l} (M-1)^l\leqslant M^{n-k},
	\end{equation}
	so that
	\begin{equation}
		H_M(\delta)\leqslant 1-\frac{k}{n}+\frac{\log_M(n+1)}{n}.
	\end{equation}
	Note that $H_M(\delta)$	is strictly increasing for $\delta \in [0, \frac{M-1}{M}]$, which implies $H_M^{-1}(x)$ is strictly increasing for $x \in [H_M(0), H_M\left(\frac{M-1}{M}\right)]$. Since $\delta<\frac{1}{2}$ and $M \geqslant 2$, we are always in the increasing region. Thus:
	\begin{align}
		&\delta=\frac{d\subscript{min}\superscript{max}(n)-1}{2n} \leqslant H_M^{-1}\left[1-\frac{k}{n}+\frac{\log_M(n+1)}{n}\right],\\
		&\frac{d\subscript{min}\superscript{max}(n)}{n} \leqslant 2H_M^{-1}\left[1-\frac{k-\log_M(n+1)}{n}\right]+\frac{1}{n}\label{eq:aux_bound}.
	\end{align}
	Note that $H_M^{-1}(1)=\frac{M-1}{M}$, the lemma is therefore proven.

\end{proof}

\section{Proof of Lemma~\ref{lem:bound_lambda}}\label{app:bound_lambda}
\begin{proof}
	Recall Eq.~\eqref{eq:M-ary_entropy}, we have
	\begin{equation}
		\begin{split}
			H_M\left(\frac{1}{2M}\right)=&\frac{1}{2M}\log_M[2M(M-1)]\\
			&-\left(1-\frac{1}{2M}\right)\log_M\left(1-\frac{1}{2M}\right)\\
			>&\frac{\log_M(2M^2)}{2M}>\frac{1}{k+1}.
		\end{split}
	\end{equation}
	Hence, $\lambda=2H_M^{-1}\left(\frac{1}{k+1}\right)<\frac{1}{M}$
\end{proof}

\section{Proof of Theorem~\ref{th:monotonicity_PconLB_regarding_n}}\label{app:monotonicity_PconLB_regarding_n}
\begin{proof}
	First, note that given certain $D$, $P\subscript{pair}$ monotonically decreases with increasing $n$, here we write it as $P\subscript{pair}^{(n)}$. More specifically, since it always holds $\theta_w(D,R)<\frac{\pi}{2}$ for $D\geqslant 2R$, the decay rate of $P\subscript{pair}^{(n)}(D)$ is exponential regarding $n$.

	Thus, within every $d\subscript{min}\superscript{max}$-constant interval of $n$, i.e. when $d\subscript{min}\superscript{max}(n+1)=d\subscript{min}\superscript{max}(n)$, we have
	\begin{equation}
		\begin{split}
			&\frac{P\subscript{con}\superscript{LB}(n+1)}{P\subscript{con}\superscript{LB}(n)}\\
			=&\frac{
			\comb{n+1}{d\subscript{min}\superscript{max}(n)}\frac{(M-1)^{d\subscript{min}\superscript{max}(n)}}{M^{n+1-k}}P\subscript{pair}^{(n+1)}\left(\sqrt{E d\subscript{min}\superscript{max}}\right)
			}{
			\comb{n}{d\subscript{min}\superscript{max}(n)}\frac{(M-1)^{d\subscript{min}\superscript{max}(n)}}{M^{n-k}}P\subscript{pair}^{(n)}\left(\sqrt{E d\subscript{min}\superscript{max}}\right)}\\
			=&\underset{\in(0,1)\text{ (Lemma~\ref{lem:bound_lambda})}}{\underbrace{\frac{(n+1)}{[n+1-d\subscript{min}\superscript{max}(n)]M}}}
			\cdot\underset{\in(0,1)}{\underbrace{\frac{P\subscript{pair}^{(n+1)}\left(Ed\subscript{min}\superscript{max}\right)}{P\subscript{pair}^{(n)}\left(Ed\subscript{min}\superscript{max}\right)}}}\in(0,1)
		\end{split}
		\label{eq:ratio_PconLB_n+1_n}
	\end{equation}
	Now investigate the $n$ instances where $d\subscript{min}\superscript{max}$ increases. From Eq.~\eqref{eq:upper_bound_dmin} it implies $d\subscript{min}\superscript{max}(n+1)=d\subscript{min}\superscript{max}(n)+2$, so:
	\begin{equation}
		\begin{split}
			&\frac{P\subscript{con}\superscript{LB}(n+1)}{P\subscript{con}\superscript{LB}(n)}\\
			=&\frac{
			\comb{n+1}{d\subscript{min}\superscript{max}(n)+2}\frac{(M-1)^{d\subscript{min}\superscript{max}(n)+2}}{M^{n+1-k}}
			}{
			\comb{n}{d\subscript{min}\superscript{max}(n)}\frac{(M-1)^{d\subscript{min}\superscript{max}(n)}}{M^{n-k}}
			}\\
			&\cdot\frac{P\subscript{pair}^{(n+1)}\left(\sqrt{E(d\subscript{min}\superscript{max}(n)+2)}\right)
			}{P\subscript{pair}^{(n)}\left(\sqrt{Ed\subscript{min}\superscript{max}(n)}\right)}\\
			=&\overset{\in\mathcal{O}(n^\rho),~0<\rho\leqslant 2}{\overbrace{\frac{[n+1][n-d\subscript{min}\superscript{max}(n)]}{[d\subscript{min}\superscript{max}(n)+1][d\subscript{min}\superscript{max}(n)+2]}}}\cdot\overset{n\text{-independent}}{\overbrace{\frac{(M-1)^2}{M}}}\\
			&\cdot\underset{\in\mathcal{O}(e^{-\alpha n}),~\alpha>0}{\underbrace{\frac{P\subscript{pair}^{(n+1)}\left(\sqrt{E(d\subscript{min}\superscript{max}(n)+2)}\right)}{P\subscript{pair}^{(n)}\left(\sqrt{E(d\subscript{min}\superscript{max}(n))}\right)}}},
		\end{split}
	\end{equation}
	which reveals that with enough large $n$, $P\subscript{con}\superscript{LB}(n+1)<P\subscript{con}\superscript{LB}(n)$ is ensured, and $P\subscript{con}\superscript{LB}$ becomes therefore monotonic.
\end{proof}

\section{Proof of Theorem~\ref{th:Pcon_UB_regarding_E}}\label{app:Pcon_UB_regarding_E}
\begin{proof}
	Recall the definition Eq.~\eqref{eq:upper_bound_Pcon}, fixing $\left(M, k, d\subscript{min}\superscript{min}\right)$, it is trivial to see that $P\subscript{con}\superscript{UB}$ has jump discontinuities at the jump discontinuities of $d\subscript{min}\superscript{min}$, and piecewise continuous in between. Especially, these jump discountinuities are at $E_i=4R^2(\varepsilon)/i$ for $i\in\mathbb{N}^+$, where we have 
	\begin{equation}
		d\subscript{min}\superscript{min}(E_i)=\left\lceil\frac{4R^2(\varepsilon)}{4R^2(\varepsilon)/i}\right\rceil=i.
	\end{equation}
	Moreover, in each of the intervals $[E_i,E_{i+1})$, its first-order partial derivative regarding $E$ is
	\begin{equation}
		\begin{split}
			&\frac{\partial P\subscript{con}\superscript{UB}}{\partial E}=\frac{\partial}{\partial E}\left[(M^k-1)P\subscript{pair}\left(\sqrt{Ed\subscript{min}\superscript{min}}\right)\right]\\
			=&\frac{M^k-1}{2}\sqrt{\frac{d\subscript{min}\superscript{min}}{E}}\cdot\left.\frac{\partial}{\partial D}P\subscript{pair}(D)\right\vert_{D=\sqrt{Ed\subscript{min}\superscript{min}}}<0,
		\end{split}
		\label{eq:parital_Pcon_upper_bound_against_E}
	\end{equation}
	which implies that $P\subscript{con}\superscript{UB}$ monotonically decreases w.r.t. $E$ within every individual interval of constant $d\subscript{min}\superscript{min}$.
\end{proof}

\section{Proof of Corollary~\ref{co:Pcon_UB_regarding_E_extremes}}\label{app:Pcon_UB_regarding_E_extremes}
\begin{proof}
	From Theorem~\ref{th:Pcon_UB_regarding_E} and Eq.~\eqref{eq:lower_bound_dmin}, it is trivial to see that $P\subscript{con}\superscript{UB}$ is piecewise decreasing within every interval $[E_i,E_{i+1})$ where $i\in\mathbb{N}^+$.
	
	Now we focus on the left and right limits at the jump discontinuities. Define $E_i^-\triangleq \lim\limits_{\delta\to 0} E_i-\vert\delta\vert$, we have $d\subscript{min}\superscript{min}(E_i^-)=i+1$. Moreover, since $d\subscript{min}\superscript{min}=i$ and $E_{i+1}<E_{i}$, it holds that 
	\begin{equation}
	\begin{split}
		&P\subscript{con}\superscript{UB}(E_i^-)\\
		=&(M^k-1)\lim_{\delta\to 0^+}P\subscript{pair}\left(\sqrt{(E_i-\delta)(i+1)}\right)\nonumber\\
		=&(M^k-1)P\subscript{pair}\left(\sqrt{E_i(i+1)}\right)\\
		<&(M^k-1)P\subscript{pair}\left(\sqrt{E_i\cdot i}\right)
		=P\subscript{con}\superscript{UB}(E_i).\nonumber
	\end{split}
	\end{equation}
	which makes $P\subscript{con}\superscript{UB}(E_i)$ local maximums and $P\subscript{con}\superscript{UB}(E_i^-)$ local minimums.
	Furthermore, for all $i\in\mathbb{N}^+$, we have
	\begin{equation}
		E_id\subscript{min}\superscript{min}(E_i)=4R^2(\varepsilon)=E_{i+1}d\subscript{min}\superscript{min}(E_{i+1}),
	\end{equation}
	which reveals that all the local maximums are equal.
	Meanwhile, since
	\begin{equation}
		\begin{split}
			&\lim_{\delta\to 0^+}E_i^-d\subscript{min}\superscript{min}(E_i^-)=4R^2(\varepsilon)\cdot\frac{i+1}{i}\\
			>&4R^2(\varepsilon)\cdot\frac{i+2}{i+1}=\lim_{\delta\to 0^+}E_{i+1}^-d\subscript{min}\superscript{min}(E_{i+1}^-),
		\end{split}
	\end{equation}
	we know that
	\begin{equation}
		\begin{split}
			&\lim_{\delta\to 0^+}(M^k-1)P\subscript{pair}\left(\sqrt{E_i^-d\subscript{min}\superscript{min}(E_i^-)}\right)\\
			<&\lim_{\delta\to 0^+}(M^k-1)P\subscript{pair}\left(\sqrt{E_{i+1}^-d\subscript{min}\superscript{min}(E_{i+1}^-)}\right),
		\end{split}
	\end{equation}
	which implies that the local minimums are monotonically decreasing regarding $i$.
\end{proof}

\clearpage
\bibliographystyle{IEEEtran}
\bibliography{references}

\end{document}